\documentclass{amsart}
\usepackage{amsfonts}

\usepackage{amsmath}
\usepackage{graphicx}

\newtheorem{theorem}{Theorem}
\theoremstyle{plain}

\newtheorem{definition}{Definition}

\newtheorem{lemma}{Lemma}

\newtheorem{remark}{Remark}

\numberwithin{equation}{section}

\begin{document}
\title[A discrete generalization of Hill's statistic.]{On a discrete Hill's statistical process based on sum-product statistics and
its finite-dimensional asymptotic theory.}
\author{Gane Samb Lo.}
\address{Universit\'{e} Gaston Berger de Saint-Louis, LERSTAD, Bp 234, Saint-Louis,
S\'{e}n\'{e}gal.}
\email{ganesamblo@ufrsat.org}

\begin{abstract}
The following class of sum-product statistics 
\begin{equation*}
T_{n}(p)=\frac{1}{k}\overset{}{\overset{p}{\underset{h=1}{\sum }}\underset{%
(s_{1}......s_{h})\in P(p,h)\text{ \ }}{\sum \text{ \ \ }}\overset{i_{0}}{%
\underset{i_{1}=l+1}{\sum ...}}\overset{}{\underset{}{\overset{i_{h-1}}{%
\underset{i_{h}=l+1}{\sum }}}\text{\ }i_{h}\overset{i_{h}}{\underset{i=i_{1}%
}{\Pi }}}\overset{}{\frac{\left( Y_{n-i+1,n}-Y_{n-i,n}\right) ^{s_{i}}}{%
s_{i}!}}}
\end{equation*}
(where $l,$ $k=i_{0}$ and n are positive integers, $0<l<k<n,$ $P(p,h)$ is
the set of all ordered parititions of $\ p>0$ into $\ h$ positive integers
and $Y_{1,n}\leq ...\leq Y_{n,n}$ are the order statistics based on a
sequence of independent random variables $Y_{1},$ $Y_{2},...$with underlying
distribution $\mathbb{P}(Y\leq y)=G(Y)=F(e^{y})$), is introduced. For each
p, $T_{n}(p)^{-1/p}$ is an estimator of the index of a distribution whose
upper tail varies regularly at infinity. \ This family generalizes the so
called Hill statistic and the Dekkers-Einmahl-De Haan one. We study the
limiting laws of the process $\left\{ T_{n}(p),1\leq p<\infty \right\} $ and
completely describe the covariance function of the Gaussian limiting process
with the help \ of combinatorial techniques. Many results available for
Hill's statistic regarding asymptotic normality and laws of the iterated
logarithm are extended to each margin $T_{n}(p,k)$, for $p$ fixed, and for
any distribution function lying in the extremal domain. In the process, we
obtain special classes of numbers related to those of paths joining the
opposite coins within a parallelogram.
\end{abstract}

\maketitle

\Large

\section{\protect\bigskip INTRODUCTION}

Let $X_{1},$ $X_{2},...$ be \ a sequence of independent copies of a random
variable $\left( r.v.\right) $ with $\mathbb{P}\left( X_{_{i}}\leq x\right)
=F(x),x\in \mathbb{R}$ and $F(1)=0$ and let $X_{1,n}\leq ...\leq X_{n,n}$
denote the corresponding order \ statistics. Hill (1975) introduced

\begin{equation}
T_{n}(1,k,l)=\frac{1}{k}\text{ }\overset{}{\underset{l+1\leq j\leq k}{\sum }%
\text{ }j(\log X_{n-j+1,n}-\log X_{n-j,n})}  \label{alpha1}
\end{equation}
(where $l,$ $k$ and $n$ are integers such that 0$\leq l=l(n)<k=k(n)<n,$ $%
l/k\rightarrow 0,k\rightarrow \infty ,$ $k/n\rightarrow 0,$ log stands for
the Neperian logarithm) as an estimator for the exponent $\gamma ^{-1}>0$ \
of a distribution function $(d.f.)$ whose the upper tail varies regulary at
infinity, that is :

\begin{equation}
\forall \lambda >0,\underset{x\rightarrow +\infty }{\lim }\text{ }%
(1-F(\lambda x)\text{ }/\text{ }(1-F(x)=\lambda ^{-\gamma }  \label{alpha2}
\end{equation}
Mason in \cite{mason} proved that the strong or weak convergence of $%
T_{n}(1,k,1)$ $\ $for $k=\left[ n^{\alpha }\right] ,$ $0<\alpha <1,$ to some 
$\gamma ^{-1}>0$ for one value of $\alpha ,0<\alpha <1,$ is characteristic
of $\ d.f.$'s statisfying (\ref{alpha1}). This means in fact that $F$
belongs to the extremal domain of attraction of a Fr\'{e}chet's law of
exposent $\gamma ^{-1}$. It is well-known in extreme value theory that a $%
d.f.$ $F$ is attracted to some non-degenerate $d.f.M$, that is there exists
two sequences of real numbers $\left( a_{n}>0\right) _{n\geq 0}$ and $\left(
b_{n}\right) _{n\geq 0}$ such that.

\begin{equation}
\forall x\in\mathbb{R},\text{ \ \ }\underset{x\rightarrow+\infty}{\lim}\text{
}F^{n}(a_{n}x+b_{n})=M(x),  \label{alpha3}
\end{equation}
iff $M$ is the Fr\'{e}chef type of $d.f.$ of parameter $\gamma>0$, 
\begin{equation}
\varphi_{\gamma}(x)=\left\{ 
\begin{array}{cc}
\exp(-x^{-\gamma}) & x\geq0 \\ 
0 & elsewhere
\end{array}
\right. ,  \label{alpha3a}
\end{equation}
or the Weibull type of $d.f.$ of parameter $\gamma>0,$

\begin{equation}
\psi_{\gamma}(x)=\left\{ 
\begin{array}{cc}
\exp(-(-x)^{\gamma}) & x\leq0 \\ 
1 & elsewhere
\end{array}
\right.  \label{alpha4}
\end{equation}
or the Gumbel type of $d.f.,$

\begin{equation}
\Lambda (x)=\exp (-e^{-x}),\text{ }x\in \mathbb{R}  \label{alpha5}
\end{equation}
A huge amount of statistical applications for modeling extremes ares based
on these limiting laws. For instance, Resnick (1987) cited a wide range of
phenomena described by extreme value theory. Further, in order to
characterize the d.f.'s attracted to one of these three d.f.'s, L\^{o} (see 
\cite{gslo3}) introduced

\begin{equation}
T_{n}(2,k)=\frac{1}{k}\text{ }\overset{i=k}{\underset{i=l+1}{\sum }}\text{\ }%
i(1-\delta _{ij}/2)(\log X_{n-i+1,n}-\log X_{n-i,n})(\log X_{n-j+1,n}-\log
X_{n-j,n})  \label{alpha6}
\end{equation}
and showed that the couple $\left( T_{n}(1,k,l),T_{n}(2,k,l)\right) $
combined with some other auxilliary statistics characterizes the whole
domain of attraction and each of the three domains of attraction. These
characterizing statistics are studied in details in \cite{gslo1} and \cite
{gslo2}, corresponding laws of the iterated logarithm are given. $%
T_{n}(2,k,l)$ may be viewed as a second dorder form of $T_{n}(1,k,l).$ In
the same view, Dekkers, Einmahl and De Haan introduced a moment estimator in 
\cite{dedh} 
\begin{equation}
A_{n}=\frac{1}{k}\text{ }\overset{i=k}{\underset{i=1}{\sum }}\text{\ }i(\log
X_{n-i+1,n}-\log X_{n-i,n})^{2}  \label{alpha6a}
\end{equation}
And it happens that T$_{n}(2,k,1)=2A_{n.}$ The importance of this second
order estamator is that the couple $\left( T_{n}(1,k,l),T_{n}(2,k,l)\right) $
separates the whole domain of attaction in the sens that the couple $\left(
T_{n}(1,k,l),\left( T_{n}(1,k,l)T_{n}(2,k,l)^{-1/2}\right) \right) $ takes
three limiting values according to three the possibilities : F$\in D(\phi
_{\gamma }),$ F$\in D(\psi _{\gamma }),$ F$\in D(\Lambda ).$ Theses previous
works of many authors cleary suggested extensions to higher orders and
advocated a stochastic process view. Our aim is to generalize results
available for $T_{n}(1,k,l)$ and $T_{n}(2,k,l)$ for higher order forms $%
T_{n}(p,k),p\geq 1,$ and to study the limiting behavior of the process $%
\left\{ T_{n}(p,k),1\leq p<+\infty \right\} .$

Let us make further notations before defining $T_{n}(p,k,l)\equiv T_{n}(k).$
Let $P(p,h)$ be the set of all ordered partitions of $p>0$ into positive
integers, $1\leq h\leq p:$

\begin{equation}
P(p,h)=\left\{ (s_{1}...s_{h}),\forall i,\text{ }1\leq i\leq
h,s_{i}>0;s_{1}+...+s_{h}=p\right\} ,1\leq h\leq p.  \label{alpha7}
\end{equation}
\qquad\qquad\qquad

Let $Y_{1,n}\leq ...\leq Y_{n,n}$ be the order statistics of the $n$ first
of a sequence of independent random variables $Y_{1},Y_{2},...$ such that $%
P(Y\leq y)=G(y)=F(e^{y}),$ $y\geq 0.$ Now define for $1\leq l<k<n,$ $p\geq
1, $\ $i_{0}=k,$%
\begin{equation}
T_{n}(p)=\frac{1}{k}\overset{}{\overset{p}{\underset{h=1}{\sum }}\underset{%
(s_{1}......s_{h})\in P(p,h)\text{ \ }}{\sum \text{ \ \ }}\overset{i_{0}}{%
\underset{i_{1}=l+1}{\sum }}...\overset{}{\underset{}{\overset{i_{h-1}}{%
\underset{i_{h}=l+1}{\sum }}}\text{\ }i_{h}\overset{i_{h}}{\underset{i=i_{1}%
}{\Pi }}}\overset{}{\frac{\left( Y_{n-i+1,n}-Y_{n-i,n}\right) ^{s_{i}}}{%
s_{i}!}}}.  \label{alpha7a}
\end{equation}

It is easy to check that $T_{n}(1)$ is Hill's statistic, $T_{n}(2)$ is that
of $(1.3)$ and, for instance, for p=3, we have

\begin{align}
T_{n}(3)& =\frac{1}{k}\overset{}{\{\sum_{j=l+1}^{j=k}\frac{j}{6}}\left(
Y_{n-j+1,n}-Y_{n-j}\right) ^{3}+\overset{}{\underset{}{\overset{i=k}{%
\underset{j=l+1}{\sum }}}\overset{}{\underset{}{\overset{i=j}{\text{ }%
\underset{i=l+1}{\sum }}}}}\frac{i}{2}\{\left( Y_{n-j+1,n}-Y_{n-j,n}\right)
^{2}  \notag \\
& \times \left( Y_{n-i+1,n}-Y_{n-i,n}\right) +(Y_{n-j+1,n}-Y_{n-j,n}\left(
Y_{n-i+1,n}-Y_{n-i,n}\right) ^{2}  \notag \\
& +\overset{i=j}{\underset{i=l+1\text{ \ }}{\sum }}\overset{m=i}{\underset{%
m=l+1}{\sum }}m\left( Y_{n-m+1,n}-Y_{n-m,n}\right) \left(
Y_{n-i+1,n}-Y_{n-i,n}\right) \times  \notag \\
& (Y_{n-j+1,n}-Y_{n-j,n})\}
\end{align}

\bigskip \qquad

\qquad Our aim in this paper is to establish the asymptotic normality theory
for the finite-dimensional distribution of this process. Our best
achievement is the convergence of the finite-dimensional distributions of
the process $T_{n}=\left( T_{n}(p),1\leq p<\infty \right) $ when suitably
centered and normalized to those of a Gaussian process $\left\{ I(p),1\leq
p<\infty \right\} .$ The computation the covariance function of this
Gaussian process required much combinatorial calculations. This
combinatorial work, along with the study of the numbers classes which appear
for the varaince calculation, is the one of the main part of this paper
since the results to be used for the $d.f.$'s $F$ are largely developped in 
\cite{gslo1}, \cite{gslo2} and \cite{gslo3} and other authors.

First, we state limit theorems in section II along \ with the methods to
compute the covariance function of $I(p).$ Section 3 is devoted to prove the
combinatorial relations in section 2. Finally, section 4 is devoted to the
proofs or the theorems. The reader is is referred to \cite{leadrotz}, \cite
{resnick}, \cite{kotz}, \cite{reiss}, and \cite{galambos} for general \
references on extreme values theory.

\bigskip

Before we state our results, we mention that a continuous generalization of
the Hill statistics has been given and studied (\cite{dioplo2006}).

\section{STATEMENT AND DESCRIPTION OF THE RESULTS}

We begin with the description of a class of numbers involved in this work.
We introduced three classes of numbers that will be used for the covariance
computations.

\begin{definition}
\label{def21}The non-negative integers $\beta (v,r),$ $v\geq 0,$ $r\geq 1$
defined by
\end{definition}

\begin{enumerate}
\item[i.]  $\forall v\geq 0,$ $\beta (v,1)=1$

\item[ii.]  $\forall v\geq 1,$ $\beta (v,2)=1$

\item[iii.]  $\forall r\geq 3,$ $\beta (1,r)=\beta (2,r-1)+(1,r-1)$

\item[iv.]  $\forall r\geq 2,$ $\beta (0,r)=\beta (1,r-1)$

\item[v.]  $\forall v\geq 2,$ $r\geq 3,$ $\beta (v,r)=\beta (v+1,r-1)+\beta
(v-1,r)$
\end{enumerate}

are called the type I numbers.

\qquad

\qquad We have the two following clog rules associated with points (iii) and
(v) of this definition.

\begin{center}
$
\begin{array}{cc}
\begin{tabular}[b]{|c|c|}
\hline
v & u+v \\ \hline
u &  \\ \hline
&  \\ \hline
\end{tabular}
& 
\begin{tabular}[b]{|c|c|}
\hline
& v \\ \hline
& u+v \\ \hline
u &  \\ \hline
\end{tabular}
\end{array}
$
\end{center}

The repeated application of them leads to the computation of all the numbers
of this class. For instance for $v\leq 10$, $r\leq 10$, we have the
different values of $\beta (\nu ,r)$

$
\begin{tabular}[t]{|c|c|c|c|c|c|c|c|c|c|c|}
\hline
$v\backslash r$ & $1$ & $2$ & $3$ & $4$ & $5$ & $6$ & $7$ & $8$ & $9$ & $10$
\\ \hline
$0$ & $1$ & $1$ & $1$ & $2$ & $5$ & $14$ & $42$ & $132$ & $429$ & $1430$ \\ 
\hline
$1$ & $1$ & $v=1$ & $u+v=2$ & $5$ & $14$ & $42$ & $132$ & $429$ & $1430$ & $%
3862$ \\ \hline
$2$ & $1$ & $u=1$ & $3$ & $9$ & $28$ & $90$ & $297$ & $1001$ & $2432$ & $%
7294 $ \\ \hline
$3$ & $1$ & $1$ & $4$ & $t=14$ & $48$ & $165$ & $572$ & $2002$ & $6072$ & 
\\ \hline
$4$ & $1$ & $1$ & $5$ & $t+s=20$ & $75$ & $275$ & $1001$ & $3640$ &  &  \\ 
\hline
$5$ & $1$ & $1$ & $s=6$ & $27$ & $110$ & $429$ & $1638$ &  &  &  \\ \hline
$6$ & $1$ & $1$ & $7$ & $35$ & $154$ & $637$ &  &  &  &  \\ \hline
$7$ & $1$ & $1$ & $8$ & $44$ & $208$ &  &  &  &  &  \\ \hline
$8$ & $1$ & $1$ & $9$ & $54$ &  &  &  &  &  &  \\ \hline
$9$ & $1$ & $1$ & $10$ &  &  &  &  &  &  &  \\ \hline
$10$ & $1$ & $1$ & $1$ &  &  &  &  &  &  &  \\ \hline
\end{tabular}
\ $

\begin{remark}
The second clog rule is that of the numbers $\beta ^{\ast }(v,r)$ of paths
in $\mathbb{Z}^{2}$ joining $\left( 0,0\right) $ to $\left( v+r,r\right) $
within the parallelogram $\left[ (0,0),(v,0),(v+r,r),(r,r)\right] $ which is
\end{remark}

\begin{equation}
\beta(v,r)=\left( 
\begin{array}{c}
r \\ 
v+2r
\end{array}
\right) -2\left( 
\begin{array}{c}
r-2 \\ 
v+2r
\end{array}
\right) ,r\leq v+1  \label{alpha10}
\end{equation}
(See \cite{kreweras}). The differences between $\beta^{\ast(.,.)}$ and $%
\beta(.,.)$ are the following:

\begin{enumerate}
\item  \ $(v,r)$ satisfies $r\leq v+1$ for $\beta^{\ast}(.,.)$ while $%
v\geq0,\gamma\geq2$ are arbitrary in $\beta(.,.).$

\item  The first lines $(v=0$ $and$ $v=1)$ differ for $\beta\ast(.,.)$ $and $
$\beta(.,.).$
\end{enumerate}

\begin{definition}
\label{def22}The non-negative intergers $\mu _{\tau }(0,v,\delta ),1\leq
\delta \leq \tau $ verifying
\end{definition}

\begin{enumerate}
\item[i]  $\forall $ $\delta \geq 0,$ $\ \ \forall \tau \geq 1,$ $\ \mu
_{\tau }(0,v,\tau )=1$

\item[ii]  $\forall $ $\delta ,$ $\ 1\leq \delta <\tau ,$ $\ \mu _{\tau
}(0,0,\delta )=1$

\item[iii]  $\forall $ $\delta ,$ $\ 1\leq \delta \tau ,$ $\ \forall v\geq
1, $ $\ \mu _{\tau }(0,v,\delta )=\mu _{\tau }(0,v-1,\delta )+(0,v,\delta
+1) $
\end{enumerate}

are called the $\tau -class$ of $type$ II numbers.

\begin{remark}
\bigskip Point (iii) of this definition yields the clog rule 
\begin{equation*}
\begin{tabular}[b]{|c|c|}
\hline
u &  \\ \hline
u+v & v \\ \hline
\end{tabular}
\end{equation*}
which enables to compute easily these numbers . Here are some examples.
\end{remark}

\begin{center}
$
\begin{array}{ccc}
\begin{tabular}{cc}
$\tau=1$ &  \\ \hline
\multicolumn{1}{|l}{$v\backslash\delta$} & \multicolumn{1}{|l|}{$1$} \\ 
\hline
\multicolumn{1}{|l}{$0$} & \multicolumn{1}{|l|}{$1$} \\ \hline
\multicolumn{1}{|l}{$1$} & \multicolumn{1}{|l|}{$1$} \\ \hline
\multicolumn{1}{|l}{$2$} & \multicolumn{1}{|l|}{$1$} \\ \hline
\multicolumn{1}{|l}{$3$} & \multicolumn{1}{|l|}{$1$} \\ \hline
\multicolumn{1}{|l}{$4$} & \multicolumn{1}{|l|}{$1$} \\ \hline
\multicolumn{1}{|l}{$5$} & \multicolumn{1}{|l|}{$1$} \\ \hline
\end{tabular}
& 
\begin{tabular}{ccc}
& $\tau=2$ &  \\ \hline
\multicolumn{1}{|c}{$v\backslash\delta$} & \multicolumn{1}{|c}{$1$} & 
\multicolumn{1}{|c|}{$2$} \\ \hline
\multicolumn{1}{|c}{$0$} & \multicolumn{1}{|c}{$1$} & \multicolumn{1}{|c|}{$%
1 $} \\ \hline
\multicolumn{1}{|c}{$1$} & \multicolumn{1}{|c}{$2$} & \multicolumn{1}{|c|}{$%
1 $} \\ \hline
\multicolumn{1}{|c}{$2$} & \multicolumn{1}{|c}{$3$} & \multicolumn{1}{|c|}{$%
1 $} \\ \hline
\multicolumn{1}{|c}{$3$} & \multicolumn{1}{|c}{$4$} & \multicolumn{1}{|c|}{$%
1 $} \\ \hline
\multicolumn{1}{|c}{$4$} & \multicolumn{1}{|c}{$5$} & \multicolumn{1}{|c|}{$%
1 $} \\ \hline
\multicolumn{1}{|c}{$5$} & \multicolumn{1}{|c}{$6$} & \multicolumn{1}{|c|}{$%
1 $} \\ \hline
\end{tabular}
& 
\begin{tabular}{cccc}
& $\tau=3$ &  &  \\ \hline
\multicolumn{1}{|c}{$v\backslash\delta$} & \multicolumn{1}{|c}{$1$} & 
\multicolumn{1}{|c}{$2$} & \multicolumn{1}{|c|}{$3$} \\ \hline
\multicolumn{1}{|c}{$1$} & \multicolumn{1}{|c}{$1$} & \multicolumn{1}{|c}{$1$%
} & \multicolumn{1}{|c|}{$1$} \\ \hline
\multicolumn{1}{|c}{$2$} & \multicolumn{1}{|c}{$3$} & \multicolumn{1}{|c}{$2$%
} & \multicolumn{1}{|c|}{$1$} \\ \hline
\multicolumn{1}{|c}{$3$} & \multicolumn{1}{|c}{$6$} & \multicolumn{1}{|c}{$3$%
} & \multicolumn{1}{|c|}{$1$} \\ \hline
\multicolumn{1}{|c}{$4$} & \multicolumn{1}{|c}{$10$} & \multicolumn{1}{|c}{$%
4 $} & \multicolumn{1}{|c|}{$1$} \\ \hline
\multicolumn{1}{|c}{$5$} & \multicolumn{1}{|c}{$15$} & \multicolumn{1}{|c}{$%
5 $} & \multicolumn{1}{|c|}{$1$} \\ \hline
\multicolumn{1}{|c}{$5$} & \multicolumn{1}{|c}{$21$} & \multicolumn{1}{|c}{$%
6 $} & \multicolumn{1}{|c|}{$1$} \\ \hline
\end{tabular}
\end{array}
$
\end{center}

\bigskip

Finally,

\begin{definition}
\label{def23}The $\tau -class$ of $type$ III numbers are the non-negative
integers $\mu _{\tau }(1,v,\delta ),$ $v\geq 0,$ $\tau \geq 1$ satisfying
points (i), (iii), (iv), (iv) anf (v) of definition \ref{def21} for $\delta
=r$ and, at the place of (ii),
\end{definition}

\begin{equation}
\forall (v\geq 1),\text{ \ \ }\mu _{\tau }(1,v,2)=\overset{v+1}{\underset{k=1%
}{\sum }}\mu _{\tau }(0,k,1)  \label{alpha11}
\end{equation}
Here are examples of computation

\begin{tabular}{cccccc}
&  & $\tau=1$ &  &  &  \\ \hline
\multicolumn{1}{|c}{$v\backslash\delta$} & \multicolumn{1}{|c}{$1$} & 
\multicolumn{1}{|c}{$2$} & \multicolumn{1}{|c}{$3$} & \multicolumn{1}{|c}{$4$%
} & \multicolumn{1}{|c|}{$5$} \\ \hline
\multicolumn{1}{|c}{$0$} & \multicolumn{1}{|c}{$1$} & \multicolumn{1}{|c}{$1$%
} & \multicolumn{1}{|c}{$2$} & \multicolumn{1}{|c}{$5$} & 
\multicolumn{1}{|c|}{$14$} \\ \hline
\multicolumn{1}{|c}{$2$} & \multicolumn{1}{|c}{$1$} & \multicolumn{1}{|c}{$3$%
} & \multicolumn{1}{|c}{$9$} & \multicolumn{1}{|c}{$28$} & 
\multicolumn{1}{|c|}{} \\ \hline
\multicolumn{1}{|c}{$3$} & \multicolumn{1}{|c}{$1$} & \multicolumn{1}{|c}{$4$%
} & \multicolumn{1}{|c}{$14$} & \multicolumn{1}{|c}{} & \multicolumn{1}{|c|}{
} \\ \hline
\multicolumn{1}{|c}{$4$} & \multicolumn{1}{|c}{$1$} & \multicolumn{1}{|c}{$5$%
} & \multicolumn{1}{|c}{} & \multicolumn{1}{|c}{} & \multicolumn{1}{|c|}{}
\\ \hline
\multicolumn{1}{|c}{$5$} & \multicolumn{1}{|c}{$1$} & \multicolumn{1}{|c}{}
& \multicolumn{1}{|c}{} & \multicolumn{1}{|c}{} & \multicolumn{1}{|c|}{} \\ 
\hline
\multicolumn{1}{|c}{} & \multicolumn{1}{|c}{} & \multicolumn{1}{|c}{} & 
\multicolumn{1}{|c}{} & \multicolumn{1}{|c}{} & \multicolumn{1}{|c|}{} \\ 
\hline
&  & $Tab$ & $2.5$ &  & 
\end{tabular}
\ \ \ \ \ $
\begin{tabular}{ccccccc}
&  & $\tau=2$ &  &  &  &  \\ \cline{1-6}
\multicolumn{1}{|c}{$v\backslash\delta$} & \multicolumn{1}{|c}{$1$} & 
\multicolumn{1}{|c}{$2$} & \multicolumn{1}{|c}{$3$} & \multicolumn{1}{|c}{$4$%
} & \multicolumn{1}{|c}{$5$} & \multicolumn{1}{|c}{} \\ \cline{1-6}
\multicolumn{1}{|c}{$0$} & \multicolumn{1}{|c}{$1$} & \multicolumn{1}{|c}{$1$%
} & \multicolumn{1}{|c}{$5$} & \multicolumn{1}{|c}{$14$} & 
\multicolumn{1}{|c}{$42$} & \multicolumn{1}{|c}{} \\ \cline{1-6}
\multicolumn{1}{|c}{$1$} & \multicolumn{1}{|c}{$1$} & \multicolumn{1}{|c}{$5$%
} & \multicolumn{1}{|c}{$14$} & \multicolumn{1}{|c}{$42$} & 
\multicolumn{1}{|c}{$132$} & \multicolumn{1}{|c}{} \\ \cline{1-6}
\multicolumn{1}{|c}{$2$} & \multicolumn{1}{|c}{$1$} & \multicolumn{1}{|c}{$9$%
} & \multicolumn{1}{|c}{$28$} & \multicolumn{1}{|c}{$28$} & 
\multicolumn{1}{|c}{$90$} & \multicolumn{1}{|c}{} \\ \cline{1-6}
\multicolumn{1}{|c}{$3$} & \multicolumn{1}{|c}{$1$} & \multicolumn{1}{|c}{$%
14 $} & \multicolumn{1}{|c}{$48$} & \multicolumn{1}{|c}{} & 
\multicolumn{1}{|c}{} & \multicolumn{1}{|c}{} \\ \cline{1-6}
\multicolumn{1}{|c}{$4$} & \multicolumn{1}{|c}{$1$} & \multicolumn{1}{|c}{$%
20 $} & \multicolumn{1}{|c}{} & \multicolumn{1}{|c}{} & \multicolumn{1}{|c}{}
& \multicolumn{1}{|c}{} \\ \cline{1-6}
\multicolumn{1}{|c}{$5$} & \multicolumn{1}{|c}{$1$} & \multicolumn{1}{|c}{}
& \multicolumn{1}{|c}{} & \multicolumn{1}{|c}{} & \multicolumn{1}{|c}{} & 
\multicolumn{1}{|c}{} \\ \cline{1-6}
&  & $Tab$ & $2.6$ &  &  & 
\end{tabular}
$

\bigskip

We finish by remarking that these tables are very quickly filled with
standard softwares. We are now able to describe the gaussian process
involved here.

\subsection{THE LIMITING GAUSSIAN PROCESSES}

\begin{definition}
The time series $\left\{ I(r),r=1,2,...\right\} $ is called an extremal
Gaussian process if and only if $\mathbb{E}(I(r))=0$ for $r\geq1$ and its
variance and covariance $\mathbb{E}(I(r)^{2})=\sigma^{2}(r)$ and $\mathbb{E}%
(I(r)I(\rho))=\sigma(r,\rho)$ statisfy for some positive functions $C_{1}(.)$
and $C_{2}(.,)$
\end{definition}

i) $a(0)=1,$ $a(1)=2$

ii) $\forall (r\geq 1),$ \ $\sigma ^{2}(r)=C_{1}(r)a(r)$ $and$ $a(r)=2$ $%
\overset{}{\overset{j=r}{\underset{j=1}{\sum }}\beta (1,j)\text{ }a(r-j)}$

iii) $\forall1\leq r<\rho,$ \ $\sigma(r,\rho)=C_{2}(r)$ $\overset{}{\overset{%
j=r}{\underset{j=0}{\sum}}\mu_{p-r}(11,j)(a(r-j)}$

\bigskip with by \ convention $\mu _{\tau }(1,1,1)=\mu _{\tau }(0,1,1),\mu
_{\tau }(1,1,0)=1$ for all $\tau \geq 1.$

\qquad Here again, the values of $\sigma^{2}(r)$ and $\sigma(r,\rho)$ are
easily computed for $C_{1}(r)\equiv C_{1}(r,\rho)\equiv1.$

\bigskip\ \ \ \ \ \ \ \ \ \ \ \ \ \ \ \ \ \ \ \ \ \ \ \ \ \ \ \ \ \ \ \ \ \
\ \ \ \ \ \ \ \ \ \ \ 
\begin{eqnarray*}
&& 
\begin{tabular}{|c|c|c|c|c|}
\hline
$4$ & $5$ & $11$ & $29$ & $70$ \\ \hline
$3$ & $4$ & $9$ & $20$ &  \\ \hline
$2$ & $3$ & $6$ &  &  \\ \hline
$1$ & $2$ &  &  &  \\ \hline
$r/\rho $ & $1$ & $2$ & $3$ &  \\ \hline
\end{tabular}
\\
&&\text{ \ \ \ \ \ \ \ \ }Tab\text{ }2.8
\end{eqnarray*}

\ \ \ \ \ \ \ \ \ \ \ \ \ \ \ \ \ \ \ \ \ \ \ \ \ \ \ \ \ \ \ \ \ \ \ \ \ \
\ \ \ \ \ \ \ \ \ \ \ \ \ \ \ \ \ \ \ \bigskip

This table cleary shows that this process is not not stationary since, for
instance,

\begin{equation}
\sigma(2,1)=3\neq9=\sigma(3,2)\neq29=\sigma(4,3)  \label{alpha12}
\end{equation}

\begin{definition}
The time series $\left\{ I(p)+e(p)Z,\ \ p=1,2,...\right\} ,$where $e\geq1$
is a real function and Z is a standard Gaussian $r.v.$ such that $\mathbb{E}%
(I(p)Z)=-1$ for all $p=1,2,...,$ is called a reduced extremal process..
\end{definition}

\subsection{LIMIT THEOREMS FOR $T_{n}(p).$}

\bigskip

Let $y_{0}=\sup \left\{ x,\text{ }G(x)<1\right\} ,$ 0$\leq x\leq z\leq y_{0}$
and

\begin{equation}
m_{1}(x,z)=\int_{x}^{z}(1-G(t))\text{ }dt  \label{alpha13}
\end{equation}

and for p$\geq 2$%
\begin{equation}
m_{p}(x,z)=\int_{x}^{z}\int_{y_{1}}^{z}....\int_{y_{p-1}}^{z}(1-G(t))\text{ }%
dt\text{ }dy_{1}...dy_{p-1}  \label{alpha14}
\end{equation}
with \ $m_{p}(x,y_{0})\equiv m_{p}(x),$ \ $p\geq 1.$ In the remainder of the
paper, we shall use, without any loss of generality$,$ the following
representations of the order statistics of Y : 
\begin{equation}
\{Y_{1,n}\leq ...\leq Y_{n,n},n\geq 1\}=\{G^{-1}(U_{1,n})\leq
G^{-1}(U_{2,n})\leq ....\leq G^{-1}(U_{n,n}),n\geq 1\}  \label{alpha15}
\end{equation}
where $U_{1,n}\leq ...\leq U_{n,n}$ are the order statistics of a sequence
of independent uniform $r.v$'s on $(0,1)$. Put now

\begin{equation}
x_{n}=G^{-1}(1-k/n),\widetilde{x}_{n}=G^{-1}(1-U_{k+1,n}),\text{ }%
z_{n}=G^{-1}(1-l/n),\widetilde{z}_{n}=G^{-1}(1-U_{l+1,n})  \label{alpha16}
\end{equation}
and 
\begin{equation}
\tau _{p}(x_{n},z_{n})=\frac{n}{k}\text{ }m_{p}(x_{n},z_{n}),p\geq 1
\label{alpha17}
\end{equation}

Finally, let \ $D(\varphi_{\gamma}),\gamma>0,$ be the set of $d.f.$'s $F$
satisfying \ref{alpha2}, D($\psi_{\gamma})$ the d.f.'s such that $%
x_{0}=\sup\left\{ x,F(x)<1\right\} <+\infty$ and $F(x_{0}-\frac{1}{\bullet}%
)\in D(\varphi_{\gamma})$ and $D(\Lambda)$ the set of $d.f.$'s $F$ such that

\begin{equation}
\underset{u\rightarrow0}{\lim}\text{ }\frac{F^{-1}(1-xu)-F^{-1}(1-u)}{%
F^{-1}(1-yu)-F^{-1}(1-u)}=\frac{\log x}{\log y}  \label{alpha18}
\end{equation}
for all $\ x>0,y>0,y\neq1.$

It is clear that $\Gamma=D(\Lambda)\cup D(\varphi)\cup D(\psi)$ where $%
D(\varphi)=U_{\gamma>0}D(\varphi_{\gamma})$ and $D(\psi)=U_{\gamma>0}D(%
\psi_{\gamma}),$ is the set of all d.f.'s attracted to some non degenerate
d.f. We have the first theorem limit.

\bigskip

\begin{theorem}
\label{theorem21}Let $F\in\Gamma,$ \ $l$ be fixed and $k$ satisfy
\end{theorem}

\begin{equation*}
(K)\text{ \ \ }0<k=k(n)\rightarrow+\infty,\text{ \ \ \ }k(n)/n\rightarrow 0%
\text{ \ }an\text{ }n\text{ }\rightarrow+\infty,
\end{equation*}
then $\left\{ k^{1/2}\left( T_{n}(p)-\tau_{p}(\widetilde{x}_{n})\right)
/\tau_{p}(x_{n}),1\leq p<+\infty\right\} ,$ converges in distribution to the
extremal gaussian process process in the canonical topology of $\mathbb{N}%
^{\infty}$ with

\begin{equation}
C_{1}(r)=\overset{r}{\underset{j=1}{\Pi}}\overset{}{\left\{ \frac{\gamma +j}{%
\gamma+r+j}\right\} ,r\geq1}  \label{alpha19}
\end{equation}
and

\begin{equation}
C_{2}(r,\rho)=\overset{j=r}{\underset{j=1}{\Pi}}\overset{}{\left\{ \frac{%
\gamma+j}{\gamma+\rho+j}\right\} ,r\geq1,\rho\geq1,r\neq\rho,}
\label{alpha20}
\end{equation}
for $0<\gamma\leq+\infty.$

Before giving the next theorem, recall that F$\in \Gamma $ may be
represented by constants $\ c$ and $d$ and by functions $f(u)$ and $b(u),$ $%
0<u<1,$ with $b(u)$ and $f(u)$ tending to zero as u tends to zero, through 
\begin{equation}
G^{-1}(1-u)=\log c-(\log u)/\gamma +\int_{u}^{1}b(t)t^{-1}dt,\text{ \ \ \ }%
0<u<1,\text{ \ \ \ }  \label{alpha21}
\end{equation}
\ For $F\in D(\varphi _{\gamma }),$ 
\begin{equation}
y_{0}-G^{-1}(i-u)=c(1+f(u))\text{ }u^{1/2}\exp
(\int_{u}^{1}b(t)t^{-1}dt),0<u<1  \label{alpha22}
\end{equation}
and for $F\in D(\psi _{\gamma }),$and

\begin{equation}
G^{-1}(1-u)=d-s(u)+\int_{u}^{1}s(t)t^{-1}dt),0<u<1,  \label{alpha22a}
\end{equation}
for $F\in D(\Lambda ),$ where $s(u)=c(1+f(u))\exp
(\int_{u}^{1}b(t)t^{-1}dt),0<u<1.$ $($\ref{alpha21}) and (\ref{alpha22}) are
the Karamata representations while (\ref{alpha22a}) is the de
Haan-Mason-Deheuvels one.

Replacing $\tau _{p}(\widetilde{x}_{n})$ by the non-randon sequence $\tau
_{p}(x_{n})$ \ requires regularity conditions on f(.). Such
characterizations for $p=1,2$ \ are given in L\^{o} (1991a). In fact, they
will hold again. But since we are only interested in putting reduced
processes to the fore, we can only use the simplest condition, that is : f'
has a derivative in some neighborhood of zero and 
\begin{equation}
\lim_{u\longrightarrow 0}uf^{\prime }u()=0.  \tag{RC}
\end{equation}

We are now able to formulate our second theorem.

\begin{theorem}
\label{theorem22}Let $F\in \Gamma $, $l$ be fixed, $k$ satisfy $(K)$ and $%
(RC)$ hold, then 
\begin{equation*}
\left\{ k^{1/2}(T_{n}(p)-\tau _{p}(x_{n}))/\tau _{p}(x_{n}),\text{ \ \ }%
1\leq p<+\infty \right\}
\end{equation*}
converges in distribution to the reduced form of the extremal process of
Theorem \ref{theorem21} in $\mathbb{N}^{\infty }$ induced with its canonical
topology with $e(p)=(\gamma +p)/\gamma ,0<\gamma \leq +\infty .$
\end{theorem}

\bigskip

\begin{remark}
Following the notation in L\^{o} (1990b), \ the case $0<\gamma <+\infty $
corresponds to $F\in D(\psi _{\gamma })$, $\gamma =+\infty $ means $F\in
D(\Lambda )\cup D(\varphi ).$ Unless the contrary is specified, the
corresponding values for the functions in $\gamma $ are obtained, in the
second case, by letting $\gamma \rightarrow +\infty $ \ in the first case.
For example, for $F\in D(\Lambda )\cup D(\varphi ),$ $e(p)=\lim_{\gamma
\rightarrow +\infty }(\gamma +p)/\gamma =1$.
\end{remark}

\bigskip

\begin{remark}
These two theorems prove the existence of the gaussian processes introduced
in Definitions \ref{def21} and \ref{def22} and show how they may be observed
and simulated.
\end{remark}

\bigskip

\section{Technical lemmas concerning numbers generation}

Recall that for all $j\geq 1,$ \ \ \ $m_{j}(x,y_{o})<+\infty $ for $F\in
\Gamma ,$%
\begin{equation}
h_{j}^{v}(\gamma
)=\int_{p_{0}}^{z_{n}}dp_{1}\int_{q_{0}}^{z_{n}}dp_{1}%
\int_{q1}^{z_{n}}dp_{2}\int_{q_{1}}^{z_{n}}...\int_{p_{p_{r-2}}}^{z_{n}}%
\frac{(q_{r-1}-p_{r-1})^{v}}{v!}\text{ }m_{j}(q_{r-1},z_{n})\text{ }dq_{r-1},
\notag
\end{equation}
$where$ $v\in \mathbb{N},$ $j\in \mathbb{N}^{\ast },r\geq 2,$\ $q_{0}=p_{1},$
\ \ $p_{0}=x_{n}$ $;$

\begin{equation}
\gamma_{j}^{v}(0,\delta)=\int_{x_{n}}^{z_{n}}dq_{1}%
\int_{q_{1}}^{z_{n}}dp_{2}...\int_{q_{\tau-\delta}}^{z_{n}}\frac{%
(q_{r-\delta+1}-x_{n})^{v}}{v!}\text{ }m_{j}(q_{\tau-\delta+1},z_{n})\text{ }%
dq_{\tau r-\delta+1};  \label{alpha23}
\end{equation}
and

\begin{equation}
\gamma_{j}^{v}(1,\delta)=\int_{x_{n}}^{z_{n}}dq_{1}%
\int_{q_{1}}^{z_{n}}dp_{2}...\int_{q_{\tau-1}}^{z_{n}}dq_{\tau}%
\int_{x_{n}}^{q_{\tau}}dp_{1}\int_{q_{\tau}}^{z_{n}}dq_{\tau+1}%
\int_{p_{1}}^{z_{n}}...\int_{p_{\delta-2}}^{q_{\tau+\delta-2}}dq_{\delta-1}
\label{alpha24}
\end{equation}

\begin{equation}
\int_{p_{\tau +\delta -2}}^{z_{n}}\left\{ \frac{(q_{\tau -\delta
-1}-p_{\delta -1})}{v!}\text{ }m_{j}(q_{\tau +\delta -1},z_{n})\right\}
dq_{\tau +\delta -1^{\prime }}  \label{alpha25}
\end{equation}
where $\nu \in \mathbb{N}^{{}},$ $\tau \in \mathbb{N}^{\ast },$ $\delta
=2,3,...$

\qquad The three class of special numbers given here appear when computing
these integrals.

\subsection{COMPUTATION OF h$_{j}^{\protect\nu}(r)$}

\begin{lemma}
\label{lemma31}For all $j\geq 1,$ \ \ \ $r>2,\nu \geq 0,$ the ratios 
\begin{equation}
\beta (\nu ,r,j)=h_{j}^{\nu }(r)/m_{j+\nu +2(r-1)}(x_{n},z_{n})
\end{equation}
are positive integers and depend only on $\left( \nu ,r\right) $ \ so that
\end{lemma}

\begin{equation}
\forall\nu\geq0,\text{ \ }\forall\geq2,\text{ \ }\beta(\nu,r,j)\equiv\beta
(\nu,r).  \label{alpha27}
\end{equation}

\begin{proof}
First put $r=2.$ We have

\begin{equation}
h_{j}^{\nu }(2)=\int_{x_{n}}^{z_{n}}dp_{1}\int_{p_{1}}^{z_{n}}\frac{%
(q_{1}-p_{1})^{\nu }}{\nu !}m_{j}(q_{1},z_{n})dq_{1},\text{ }\nu \geq
1,j\geq 1.  \label{alpha28}
\end{equation}
By remarking that,

\bigskip 
\begin{equation}
m_{j}(q_{1},z_{n})=-dm_{j+1}(q_{1},z_{n})\text{ }/dq_{1},  \label{alpha29}
\end{equation}
and by integrating by parts, we arrive at

\begin{equation}
h_{j}^{\nu }(2)=h_{j+1}^{\nu -1}(2),\text{ \ \ for }\nu \geq 1,\text{ }j\geq
1.  \label{alpha30}
\end{equation}
But, 
\begin{equation}
h_{j}^{0}(2)=m_{j+2}(x_{n},z_{n}),\text{ for }j\geq 1,  \label{alpha31}
\end{equation}
so that 
\begin{equation}
h_{j}^{0}(2)\text{ }/\text{ }m_{j+0+2(2-1)}(x_{n},z_{n})=\beta (0,2)=1.
\label{alpha32}
\end{equation}
By repeating (\ref{alpha30}) until its right member becomes $h_{\bullet
}^{0}(2)$ and by using (\ref{alpha31}), we get 
\begin{equation}
\forall (j\geq 2),\forall (\nu \geq 1),h_{j}^{\nu }(2)=m_{j+\nu
+2}(x_{n},z_{n})  \label{alpha32b}
\end{equation}
which proves the statements of the lemma for $r=2,$ that is

\begin{equation}
\forall (j\geq 1),\forall (\nu \geq 1),1=h_{j}^{\nu }(2)/m_{j+\nu
+2(2-1)}(x_{n},z_{n})=\beta (\nu ,2)  \label{alpha33}
\end{equation}
Now, by applying again a change of variables like (\ref{alpha29}) and by
integrating by parts in $h_{j}^{\nu }(\gamma ),$ we obtain

\begin{equation}
\forall (j\geq 1),\text{ \ \ }\forall (r\geq 3),\text{ \ }\forall \nu \geq 1,%
\text{ }h_{j}^{\nu }(r)=h_{j+1}^{\nu +1}(r-1)+h_{j+1}^{\nu -1}(r)
\label{alpha34}
\end{equation}
By assuming that the statements of the lemma hold for $r-1\geq 2,$ that is

\begin{equation}
\forall (j\geq 1),\forall (\nu \geq 1),\text{ }\beta (\nu ,r-1)=h_{j}^{\nu
}(r-1)/m_{j+\nu +2(r-2)}(x_{n},z_{n}),  \label{alpha35}
\end{equation}
and by repeating (\ref{alpha34}) until the second term of its right member
becomes an $h_{\bullet }^{0}(\gamma )$ term, we show that the expression $%
h_{j}^{\nu }(r)/m_{j+\nu +(r-1)}(x_{n},z_{n})=\beta (\nu ,\gamma )$ does not
depend on $\ j$ \ for any $\nu \geq 1$ and that

\begin{equation}
\forall (r\geq 1),(\forall \nu \geq 1),\text{ }\beta (\nu ,r-1)=\overset{%
h=\nu -1}{\underset{h=-1}{\sum }}\beta (\nu -h,r-1).  \label{alpha36}
\end{equation}
This proves that $\beta (\nu ,r)$ is also integer for $\nu \geq 0$ and $%
t\geq 3.$ This together \ with (\ref{alpha33}) proves lemma by induction. It
is easy to derive from this latter that 
\begin{equation}
\forall (r\geq 3),\forall (\nu \geq 2),\text{ }\beta (\nu ,r)=\beta (\nu
+1,r-1)+\beta (\nu -1,r).  \label{alpha36a}
\end{equation}
Further, for $r\geq 3,$ $\ \ j\geq 1,$ \ we obviously have

\begin{equation}
h_{j}^{0}(\gamma )=h_{j+1}^{1}(r-1)  \label{alpha37}
\end{equation}
so that

\begin{equation}
\forall (r\geq 3),\text{ }\forall j\geq 1,\text{ }\beta (0,r)=\beta (1,r-1)
\label{alpha38}
\end{equation}
By applying (\ref{alpha29}) and by integrating by parts, we arrive at

\begin{equation}
\forall (r\geq 3),\text{ }h_{j}^{1}(r)=h_{j+1}^{2}((r-1)+h_{j+1}^{0}(\gamma
),  \label{alpha39}
\end{equation}
which, combined with (\ref{alpha37}) and (\ref{alpha38}), implies

\begin{equation}
\forall (r\geq 3),\text{ }\beta (1,r)=\beta (2,r-1)+\beta (1,r-1)
\label{alpha40}
\end{equation}
Finally, set by convention

\begin{equation}
\forall (\nu \geq 0),\text{ }\beta (\nu ,1)=1  \label{alpha41}
\end{equation}
and see that (\ref{alpha36a}), (\ref{alpha38}), (\ref{alpha40}) and (\ref
{alpha41}) together show that the beta nimbers are generated by the ratios

\begin{equation}
h_{j}^{\nu }(r)/m_{j+\nu +2(r-1)}(x_{n},z_{n}),\text{ \ }for\text{ }j\geq 1,%
\text{ }\nu \geq 0,r\geq 1.  \label{alpha42}
\end{equation}
\end{proof}

\subsection{\protect\bigskip COMPUTATION OF $\protect\gamma_{J}^{\protect\nu%
}(0,\protect\delta)$}

\begin{lemma}
\label{lemma32}For all $\delta ,$ $1\leq \delta \leq \tau ,\nu \geq 0,$ the
ratios
\end{lemma}

\begin{equation}
\gamma _{j}^{\nu }(0,\tau )/m_{j+\nu +\tau -\delta +1}(x_{n},z_{n})=\mu
_{\tau }(0,\nu ,\delta ,j)\text{ }  \label{alpha43}
\end{equation}
are integers depending only on $\tau ,\nu $ and $\delta $ so that

\begin{equation}
\forall \tau \geq 1,\forall \delta ,0<\delta <\tau ,\forall \nu \geq
0,\forall j\geq 1,\text{ }\mu _{\tau }(0,\nu ,\delta ,j)\equiv \mu _{\tau
}(0,\nu ,\delta )\text{ }  \label{alpha44}
\end{equation}

\begin{proof}
A (\ref{alpha29}) - like change of variables yields

\begin{equation}
\gamma _{j}^{\nu }(0,\tau )=\gamma _{j+1}^{\nu -1}(0,\tau ),\text{ }j\geq 1.
\label{alpha45}
\end{equation}
It is easily checked that

\begin{equation}
\gamma _{j}^{0}(0,\delta )=m_{j+\tau -\delta +1}(x_{n},z_{n}),
\label{alpha46}
\end{equation}
which gives

\begin{equation}
\forall \text{ }\delta ,1\leq \delta \leq \tau ,\text{ }\mu _{\tau
}(0,0,\delta )=1  \label{alpha47}
\end{equation}
Combining this with (\ref{alpha45}), we get

\begin{equation}
\forall \nu \geq 0,\mu _{\tau }(0,\nu ,\tau )=1  \label{alpha48}
\end{equation}
By using also a (\ref{alpha29})-like change of variables, we have

\begin{equation}
\forall \text{ }j\geq 1,\forall (\nu \geq 0),\text{ }\forall (1\leq \delta
\leq \tau ),\text{ \ }\gamma _{j}^{\nu }(0,\delta )=\gamma _{j+1}^{\nu
}(0,\delta +1)+\gamma _{j+1}^{\nu -1}(0,\delta )  \label{alpha49}
\end{equation}
One proves this lemma by induction over $\tau -r,r=0,1,...,\tau -1$ through (%
\ref{alpha49}) after having taken into account the initial column given in (%
\ref{alpha48}). This induction yields

\begin{equation}
\forall \nu \geq 0,\forall 1\leq \delta <\tau ,\mu _{\tau }(0,\nu ,\delta )=%
\overset{k=\nu }{\underset{k=0}{\sum }},\mu _{\tau }(0,\nu ,\tau ),
\label{alpha50}
\end{equation}
which, in turn, gives the following clog-rule

\begin{equation}
\forall \nu \geq 1,\forall 1\leq \delta <\tau ,\text{ \ }\mu _{\tau }(0,\nu
,\delta )=\mu _{\tau }(0,\nu -1,\delta )+\mu _{\tau }(0,\nu ,\delta +1)
\label{alpha51}
\end{equation}
It is now proved that the $\tau -class$ of numbers are generated by the
integrals $\gamma _{j}^{\nu }(0,\delta ).$
\end{proof}

\subsection{ COMPUTATION OF $\protect\gamma_{j}^{\protect\nu}(0,\protect%
\delta),\protect\delta\geq2.$}

\begin{lemma}
\label{lemma33}Let \ $\tau \in \mathbb{N}^{\ast }$ be fixed. For all $j\geq
1,$ \ $\nu \geq 0$ and $\delta \geq 2,$ the ratios
\end{lemma}

\begin{equation}
\gamma_{j}^{\nu}(1,\delta)/m_{j+\nu+\tau+2(\delta-1)}(x_{n},z_{n})=\mu_{\tau
}(1,\nu,\delta,j).  \label{alpha52}
\end{equation}
are integers depending only on $(\nu,\delta)$ so that

\begin{equation}
\forall \tau \in \mathbb{N}^{\ast },\text{ }\forall (j\geq 1),\ \forall (\nu
\geq 0),\text{ }\forall (\delta \geq 2)\text{, \ }\mu _{\tau }(1,\nu ,\delta
,j)\equiv \mu _{\tau }(1,\nu ,\delta )  \label{alpha53}
\end{equation}

\begin{proof}
:\bigskip First, we compute $\gamma _{j}^{\nu }(1,2).$ A (\ref{alpha29}%
)-like change of variables gives

\begin{equation}
\gamma _{j}^{\nu }(1,2)=\gamma _{j+1}^{\nu +1}(0,1)+\gamma _{j+1}^{\nu
-1}(1,2),\text{ \ \ }j\geq 1,\nu \geq 0  \label{alpha54}
\end{equation}
By repeating this latter $\nu $ times, and by using

\begin{equation}
\forall j\geq 1,\text{ \ \ \ }\gamma _{j}^{0}(1,2)=\gamma _{j+1}^{1}(0,1),
\label{alpha55}
\end{equation}
and, finally by applying Lemma \ref{lemma32}, we get

\begin{equation}
\text{\ }\gamma _{j}^{\nu }(1,2)=(\overset{k=\nu +1}{\underset{k=1}{\sum }}%
\mu _{\tau }\left( 0,k,1\right) \text{ }m_{j+\nu +\tau +2}\left(
x_{n},z_{n}\right) .\text{ \ \ }  \label{alpha56}
\end{equation}
This proves that the statements of \ the lemma hold for $\delta =2$ with

\begin{equation}
\forall \nu \geq 0,\text{ }\mu _{\tau }\left( 1,\nu ,2\right) =\overset{\nu
+1}{\underset{k=1}{\sum }}\mu _{\tau }\left( 0,k,1\right)  \label{alpha57}
\end{equation}
From now, this proof follows the lines of that of Lemma \ref{lemma31} with
exactly the same methods. We also get the same conclusions.
\end{proof}

\bigskip

\section{PROOFS OF THEOREMS.}

Cs\"{o}rg\"{o}-Cs\"{o}rg\"{o}$-$Horv\`{a}th and Mason have constructed,in 
\cite{cchm}, a probability space $(\Omega ,\mathcal{U},\mathbb{P)}$ carrying
a sequence of Brownian \ bridges $\left\{ B_{n}(s),0\leq s\leq 1\right\} ,$ $%
n=1,2....$ and a sequence of independent uniform $r.v.^{\prime }s$ on $(0,1)$
$U_{1},U_{2},...$ such that for all $0<\nu <1/4,$

\begin{equation}
\underset{\frac{1}{n}\leq s\leq 1-\frac{1}{n}}{\sup }\frac{\left| \sqrt{n}%
(U_{n}(s)-s)-B_{n}(s)\right| }{\left( s(1-s)\right) \frac{1}{2}-\nu }%
=0_{p}(n^{-\nu })  \label{alpha58b}
\end{equation}
\ and 
\begin{equation}
\underset{\frac{1}{n}\leq s\leq 1-\frac{1}{n}}{\sup }\frac{\left| \sqrt{n}%
(s-V_{n}(s))-B_{n}(s)\right| }{\left( s(1-s)\right) \frac{1}{2}-\nu }%
=0_{p}(n^{-\nu })  \label{alpha58a}
\end{equation}
where

\begin{equation*}
U_{n}(s)=\frac{j}{n}\text{ \ }for\text{ \ }U_{j,n}\leq s<U_{j+1,n},\text{ }%
0\leq s\leq 1,
\end{equation*}
is the uniform empirical distribution function and

\begin{equation*}
V_{n}(s)=\left\{ 
\begin{array}{c}
U_{j,n}\text{ \ \ }for\text{ \ \ \ }\frac{j-1}{n}<s\leq \frac{j}{n},1\leq
j\leq n,\text{ }0<s\leq 1 \\ 
V_{n}(0)=U_{1,n}
\end{array}
\right.
\end{equation*}
is the uniform quantile function and finally, and 
\begin{equation*}
0=U_{0,n}\leq U_{1,n}\leq ...\leq U_{n,n}\leq U_{n+1,n}=1
\end{equation*}
are the order statistic of \ $U_{1},...U_{n}.$

Throughout these proofs, we suppose that we are on this probability space.
Consequently, the sequence $Y_{1},Y_{2},...$ defined above and the sequence
of empirical distribution function based on them, will be represented as

\begin{equation}
\left\{ Y_{i,n},1\leq i\leq n,n\geq1\right\} =\left\{
G^{-1}(1-U_{n-i+1,n}),1\leq i\leq n,n\geq1\right\}  \label{alpha59}
\end{equation}
and

\begin{equation}
\left\{ 1-G_{n}(x),-\infty<x<x+\infty,n\geq1\right\} =\left\{
U_{n}(1-G(x)),-\infty<x<+\infty,n\geq1\right\}  \label{alpha60}
\end{equation}
First, routine calculations yield

\begin{equation}
T_{n}(p)=\frac{n}{k}\int_{\widetilde{x}_{n}}^{\widetilde{z}%
_{n}}\int_{y_{1}}^{\widetilde{z}_{n}}...\int_{y_{p-1}}^{\widetilde{z}%
_{n}}U_{n}\left( 1-G(t)\right) dt\text{ }dy_{1}...dy_{p-1},p\geq1
\label{alpha61}
\end{equation}

The details of the computations are omitted. The reader may verify it for $%
p=1,$ $2$ and $3.$ Now, let

\begin{equation}
\alpha_{n}(s)=n^{1/2}\left( U_{n}(s)-s\right) ,0\leq s\leq1;
\label{alpha62a}
\end{equation}

\begin{equation}
m_{1,p}(\widetilde{x}_{n})=\int_{\overset{}{x}_{n}}^{\widetilde{x}_{n}}\int_{%
\overset{}{y}_{n}}^{\widetilde{x}_{n}}...\int_{y_{p-1}}^{\widetilde
{x}_{n}}B_{n}\left( 1-G(t)\right) dt\text{ \ \ }dy_{1}....dy_{p-1};
\label{alpha62b}
\end{equation}

\begin{equation}
m_{2,p}(\widetilde{z}_{n})=\int_{\overset{}{z}_{n}}^{\widetilde{x}_{n}}\int_{%
\overset{}{y}_{1}}^{\widetilde{x}_{n}}...\int_{y_{p-1}}^{\widetilde
{x}_{n}}B_{n}\left( 1-G(t)\right) dt\text{ \ \ }dy_{1}....dy_{p-1};
\label{alpha62c}
\end{equation}
\begin{equation*}
m_{3,p}(x_{n},\widetilde{z}_{n})=\int_{\overset{}{x}_{n}}^{\widetilde{z}%
_{n}}\int_{\overset{}{y}_{1}}^{\widetilde{z}_{n}}...\int_{y_{p-1}}^{%
\widetilde {z}_{n}}B_{n}\left( 1-G(t)\right) dt\text{ \ \ }%
dy_{1}....dy_{p-1};
\end{equation*}

\begin{equation}
R_{n}=\left( \frac{n}{k}\right) ^{\frac{1}{2}}\int_{\overset{}{x}_{n}}^{%
\widetilde{x}_{n}}\int_{\overset{}{y}_{1}}^{\widetilde{x}_{n}}...%
\int_{y_{p-1}}^{\widetilde{x}_{n}}\alpha_{n}\left( 1-G(t)\right)
-B_{n}\left( 1-G(t)\right) \text{ }dt\text{ \ \ }dy_{1}....dy_{p-1};
\label{alpha62d}
\end{equation}

and

\begin{equation}
W_{n}(p)=\left( \frac{n}{k}\right) ^{\frac{1}{2}}\int_{\overset{}{x}_{n}}^{%
\overset{}{z_{n}}}\int_{\overset{}{y}_{1}}^{\overset{}{z}_{n}}...%
\int_{y_{p-1}}^{\overset{}{z_{n}}}B_{n}\left( 1-G(t)\right) \text{ }dt\text{
\ \ }dy_{1}....dy_{p-1}  \label{alpha62e}
\end{equation}

We have

\begin{lemma}
\label{lemma41}For all $p\geq 1,$%
\begin{equation}
\bigskip \sqrt{k}\left( T_{n}(p)-\tau _{p}(\widetilde{x}_{n},\widetilde{z}%
_{n})\right) =w_{n}(p)+R_{n}+\left( \frac{n}{k}\right) ^{\frac{1}{2}}m_{1,p}(%
\widetilde{x}_{n})  \label{alpha63}
\end{equation}

\begin{equation*}
\times\left( \frac{n}{k}\right) ^{\frac{1}{2}}\overset{j=p-1}{\underset{j=0}{%
\sum}}\frac{\left( z_{n}-x_{n}\right) ^{j}}{j\text{ }!}m_{2,p-j}(\widetilde{z%
}_{n})+\frac{1}{2}\overset{j=p-1}{\underset{j=0}{\sum}}\frac{m_{1,j}(%
\widetilde{x}_{n})}{j\text{ }!}m_{3,p-j}\left( x_{n},\widetilde{z}%
_{n}\right) .
\end{equation*}
\end{lemma}

\begin{proof}
This is straighforward.\bigskip
\end{proof}

We already have all the necessary \textit{row matierals }(to handle the
error terms of \ Lemma \ref{lemma41}\ ) in Lemmas 4.1 and 4.2 in \cite{gslo2}%
, summerized as follows.

\begin{lemma}
\label{lemma42}\bigskip let $F\in \Gamma $ with $F(1)=0$ and et $G$ be
associated with $F$ by $G(x)=F(e^{x}),x\geq 1.$ Then
\end{lemma}

\begin{itemize}
\item[1-]  $R_{p}(x,F)\sim(x_{0}-x)^{p}\left\{ \Pi_{j=1}^{j=p}(\gamma
+j)^{-1}\right\} ,$ $as$ $x\rightarrow x_{0},$ whenever $F\in
D(\psi_{\gamma})$

\item[2-]  $R_{p}(x,G)\sim R_{1}(x,G)^{p},$ $as$ $x\rightarrow y_{0},$
whenever $F\in(\Lambda)UD(\varphi)$

\item[3-]  $((z-x)^{p}/$ $R_{p}(x,G)\rightarrow +\infty $ $and$ $%
R_{p}(x,z,G) $ $/$ $R_{p}(x,G)\rightarrow 1,$ as $x\rightarrow
x_{0},z\rightarrow x_{0},(1-G(z))$ $/$ $(1-G(x)\rightarrow 0.$
\end{itemize}

Much details on how using Lemma \ref{lemma42} to treat errors terms in Lemma 
\ref{lemma41} are given in the proofs \cite{gslo2}. We therefore omit them
only for sake of conciseness. We finally get.

\begin{lemma}
\label{lemma43}Let $F\in \Gamma .$ Then for all $p\geq 1,$
\end{lemma}

\begin{equation}
\sqrt{k}\left( T_{n}(p)-\tau _{p}(\widetilde{x}_{n})\right) \text{ }/\text{ }%
\tau _{p}(x_{n})=W_{n}(p)\text{ }/\text{ }\tau _{p}(x_{n})+o_{p}(1),\text{
as }n\rightarrow +\infty  \label{alpha65}
\end{equation}

\bigskip

Also, we have

\begin{lemma}
\label{lemma44}Let $F\in \Gamma $ and $(RC)$ hold. Then,
\end{lemma}

\begin{equation}
\sqrt{k}(\tau _{p}(\overset{}{x}_{n})\text{ }-\text{ }\tau _{p}(\widetilde{x}%
_{n}))\text{ }/\text{ }\tau _{p}(x_{n})=e(p)\text{ }nk^{-1/2}\left( U_{k,n}-%
\frac{k}{n}\right) +O_{p}(1)  \label{alpha66}
\end{equation}

\begin{equation*}
-e(p)\text{ }B_{n}\left( \frac{k}{n}\right) +O_{p}(1),\text{ as }%
n\rightarrow +\infty
\end{equation*}

\begin{proof}
(Outline of the proof).

Check that

\begin{equation}
\tau _{p}(x_{n})-\tau _{p}(\overset{\sim }{x}_{n})=\frac{k}{n}(x_{n}-%
\overset{\sim }{x}_{n})m_{p-1}(x_{n})+R_{n}(2)  \label{alpha67}
\end{equation}
where

\begin{equation*}
\frac{k}{n}\left| R_{n}(2)\right| \leq \left| x_{n}-\widetilde{x}_{n}\right|
^{p}\sup \left( \frac{k}{n},1-G(x_{n})\right) +\overset{p-1}{\underset{j=2}{%
\sum }}\frac{\left| x_{n}-\widetilde{x}_{n}\right| ^{j}}{j!}m_{p-j}(x_{n}).
\end{equation*}
By Lemma 3.6 in \cite{gslo1} and Lemma \ref{lemma42} below,

\begin{equation}
\sqrt{k}(\overset{}{x}_{n}-\widetilde{x}_{n})\text{ }/\text{ }R_{1}(x_{n})=-%
\text{ }\frac{\gamma +1}{\gamma }B_{n}(\frac{k}{n})+\mathcal{O}_{p}(1)\text{ 
}as\text{ }n\rightarrow +\infty  \label{alpha69}
\end{equation}

It is not difficult to show that $\sqrt{k}R_{n}(2)$ $/$ $\tau
_{p}(x_{n})\rightarrow _{p}0$ as $n\rightarrow +\infty $ by using Lemma \ref
{lemma42} and (\ref{alpha69}) and that $\frac{n}{k}\left( 1-G(x_{n})\right)
\rightarrow 1$ as $n\rightarrow +\infty $ . Thus, $\sqrt{k}(T_{n}(p)-\tau
_{p}(\widetilde{x}_{n},\widetilde{z}_{n}))$ $/$ $\tau _{p}(x_{n})$ $and$ $%
\sqrt{k}(T_{n}(p)-\tau _{p}(x_{n}))$ $/$ $\tau _{p}(x)$ behave
asymptotically as

\begin{equation}
W_{n}(p)\text{ }/\text{ }\tau _{p}(x_{n}),\text{ \ \ \ \ \ }p\geq 1,
\label{alpha70}
\end{equation}
and as

\begin{equation}
W_{n}(p)\text{ }/\text{ }\tau _{p}(x_{n})-e(p)B_{n}(\frac{k}{n})
\label{alpha71}
\end{equation}
\end{proof}

\bigskip

Let us say a few words on the finite-dimensional distributions (\textbf{f.d.d%
}) of (\ref{alpha70}) and (\ref{alpha71}). Since W$_{n}(p)$ is a multiple
Riemannian integral for each $p\geq 1$, it is clear that, for $n$ fixed, any
linear combination of $\left( W_{n}(p_{1}),...,W_{n}.(p_{j})\right) ,$ $\
for $ any fixed $j>1$, is limit everywhere of linear combinations of f.d.'s
of the Brownian bridge $B_{n}(\bullet .).$ But (\ref{alpha70}) and (\ref
{alpha71}) are well-defined normal random variables and then, their \textbf{%
f.d.d} are gaussian. We have now to verify that their limiting covariance
functions are finite.

\subsection{COMPUTATION OF THE VARIANCE OF $w_{n}(r),r\geq1.$}

Recall that h(s,t)=$\mathbb{E}\left( B_{n}(1-G(t)\text{ \ }(1-G(s)\right)
=\min1-G(t),$ \ $\left( 1-G(t)\right) -\left( 1-G(t)\right) \left(
1-G(s)\right) ,$ for$\ 0\leq s,t\leq1.$ One has

\begin{equation}
\forall (r\geq 1),\text{ \ \ }\mathbb{E}\text{ }(W_{n}(r))=0\text{ \ \ }
\label{alpha72}
\end{equation}
and

\begin{equation}
\forall r\geq 1,\text{ \ \ }\mathbb{E}(W_{n}(r)^{2})=\frac{n}{k}\int_{%
\overset{}{x}_{n}}^{\overset{}{z_{n}}}\int_{\overset{}{p}_{1}}^{\overset{}{z}%
_{n}}\int_{\overset{}{q}_{1}}^{\overset{}{z}_{n}}...\int_{p_{r-1}}^{\overset{%
}{z_{n}}}\int_{\overset{}{q}_{r-1}}^{\overset{}{z}_{n}}h(s,t)\text{ }ds\text{
}dt  \label{alpha73}
\end{equation}

\begin{equation*}
\times dp_{1}..dp_{r-1}\text{ \ }dq_{1}...\text{ }dq_{r-1}
\end{equation*}

\begin{equation}
=\frac{n}{k}\int_{x_{n}}^{\overset{}{z}_{n}}\int_{\overset{}{x}_{n}}^{%
\overset{}{z}_{n}}H(p_{1},q_{1})\text{ }dp_{1\text{ }}dq_{1}=\text{ }:\text{ 
}\sigma_{n}^{2}(r).  \label{alpha75}
\end{equation}
Since for all $(p,q)$ $(p<y_{0},$ $\ q<y_{0}),$ $H(p,q)=H(q,p),$ cutting the
integration space into $(p_{1}<q_{1})$ and $(p_{1}\leq q_{1})$ yieds

\begin{equation}
\text{ }\sigma _{n}^{2}(r)=2\frac{n}{k}\int_{x_{n}}^{\overset{}{z}%
_{n}}dp_{1}\int_{\overset{}{x}_{n}}^{\overset{}{z}_{n}}H(p_{1},q_{1})\text{ }%
dp_{1\text{ }}  \label{alpha76}
\end{equation}
One has (see \cite{cm} and \cite{gslo2})

\begin{equation}
\forall (x_{n}\leq y<z_{n}),\text{ }\int_{y}^{z_{n}}\int_{y}^{z_{n}}h(s,t)%
\text{ }ds\text{ }dt=2\left( \int_{y}^{\overset{}{z}_{n}}\int_{\overset{}{t}%
}^{\overset{}{z}_{n}}1-G(s)\text{ }dt\right) (1+r_{n})(1),  \label{alpha77}
\end{equation}
with $\left| r_{n}(1)\right| \leq 1-G(x_{n}).$ For the remainder, we shall
proceed by induction. Suppose that for $r\geq 2,$

\begin{equation}
\forall \text{ }j,(1\leq j\leq r-1),\forall x_{n}\leq y\leq
z_{n},\int_{y}^{z_{n}}dq_{1}\int_{y}^{z_{n}}dp_{1}\int_{p_{1}}^{z_{n}}dq_{2}%
\int_{q_{1}}^{z_{n}}dq2...dp_{j-1}\int_{p_{j-1}}^{z_{n}}ds.  \label{alpha78}
\end{equation}

\begin{equation*}
\int_{q_{j-1}}^{z_{n}}h(s,t)\text{ }ds\text{ }dt=2a(j)\text{ }m_{2j}(y,z_{n})%
\text{ }(1+r_{n}(j)),
\end{equation*}
where $\ \left| r_{n}(j)\right| \leq 1-G(x_{n})$ and $a(j)$ does not depend
on the size n. Now, by cutting 
\begin{equation*}
\int_{p_{j}}^{z_{n}}.dp_{j+1}
\end{equation*}
into 
\begin{equation*}
\int_{p_{j}}^{q_{j}}.dp_{j+1}+\int_{q_{j}}^{z_{n}}.dp_{j+1}
\end{equation*}
in \ (\ref{alpha76}) sequentially for $j=1,...,r-1$ and by applying (\ref
{alpha78}), we get

\begin{equation}
\sigma_{n}^{2}(r)=2\frac{n}{k}\text{ }\overset{j=r}{\underset{j=1}{\sum}}%
a(r-j)\text{ }h_{2(r-j)+1\text{ }}^{1}(j)\text{ }\left( 1+r_{n}(j,r)\right) 
\text{ }  \label{alpha79}
\end{equation}
where $\left| r_{n}(j,r)\right| \leq1-G(x_{n})$ and, by convention, $a(0)=1$
and $h_{\bullet}^{1}(1)=1\times m_{2+=1}(x_{n},z_{n})$

\bigskip

Now, by applying Lemma \ref{lemma31},

\begin{equation}
\sigma _{n}^{2}(r)=(2\frac{n}{k}\text{ }\underset{j=1}{\sum^{j=r}}%
a(r-j)\beta (1,j)\left( 1+r_{n}(r)\right) \text{ }m_{2r}(x_{n},z_{n})\text{ }
\label{alpha80}
\end{equation}
with $\left| r_{n}(r)\right| \leq 1-G(x_{n}).$ This and (\ref{alpha77})
prove that for all $r\geq 1,$

\begin{equation}
\sigma_{n}^{2}(r)=2\frac{n}{k}\text{ }a(r)\text{ }m_{2r}(x_{n},z_{n})\text{ }
\label{alpha81}
\end{equation}
Now, according to Lemma \ref{lemma42}, we get

\begin{equation}
\underset{n\rightarrow +\infty }{\lim }\text{ \ }\mathbb{E}(W_{n})\left(
(r)^{2}\text{ }/\text{ }\tau _{r}^{2}(x_{n})\right) =2C_{1}(r)\text{ }%
\overset{j=r}{\underset{j=1}{\sum }}a(r-j)\beta (1,j)\text{\ \ \ \ }
\label{alpha82}
\end{equation}
where $C_{1}(\gamma )=1$ for $F\in D(\Lambda )UD(\varphi ),$%
\begin{equation*}
C_{1}(\gamma )=\overset{r}{\underset{j=1}{\Pi }}\left\{ \frac{\gamma +j}{%
\gamma +r+j}\right\} forF\in D(\psi _{\gamma }).
\end{equation*}

\bigskip

\subsection{COMPUTATION OF THE COVARIANCE FUNCTION OF $w_{n}(.)$}

One has, $r<\rho ,$ $\tau =\rho -r,$

\begin{equation}
\mathbb{E}(W_{n}(r)\text{ }W_{n}(\rho ))=\frac{n}{k}%
\int_{x_{n}}^{z_{n}}dq_{1}\int_{q_{1}}^{z_{n}}...\int_{q_{\tau
-1}}^{z_{n}}dq_{\tau
}\int_{x_{n}}^{z_{n}}\int_{x_{1}}^{z_{n}}dp_{1}\int_{q\tau }^{z_{n}}dq_{\tau
+1}  \label{alpha83}
\end{equation}

\begin{equation*}
\int_{p_{1}}^{z_{n}}dq_{2}\int_{q\tau+_{1}}^{z_{n}}dq_{\tau+2}\int_{q_{\rho
-1}}^{z_{n}}\int_{p_{r-1}}^{z_{n}}h(s,t)\text{ }ds\text{ }dt
\end{equation*}

$\bigskip$Cutting $\int_{x_{n}}^{z_{n}}\bullet.dp_{j+1}$ $into$ $%
,\int_{x_{n}}^{q_{\tau}}\bullet.dp_{1}+\int_{q_{\tau}}^{z_{n}}\bullet.dp_{1}$
\ \ and \ $\int_{p_{j}}^{z_{n}}.dp_{j}$ $into$ $,\int_{p_{j}}^{q_{j}}.dp_{j}$
\ \ into $\int_{p_{j}}^{q_{\tau+j}}\bullet$ $dp_{j}+\int_{q_{%
\tau+j}}^{z_{n}}\bullet$ $dp_{j}$, we get

\begin{equation}
\mathbb{E}(W_{n}(r)W_{n}(\rho ))=\frac{n}{k}\left( a(r-1)\gamma
_{2(r-1)+1}^{1}(0,1)(1+r_{n}(r,1)\right) \times  \label{alpha84}
\end{equation}

\begin{equation*}
\overset{r}{\underset{j=2}{\sum }}\left( a(r-j)\gamma
_{2(r-1)+1}^{1}(1,j)(1+r_{n}(r,j)\right)
\end{equation*}
where $\left| r_{n}(r,j)\right| \leq 1-G(x_{n}),$ for $1\leq j\leq r.$ By
Lemma \ref{lemma33},

\begin{equation}
\mathbb{E}(W_{n}(\gamma )W_{n}(\rho ))=\left\{ \frac{n}{k}\overset{j=r}{%
\underset{j=1}{\sum }}a(r-j)\text{ }\mu _{\tau }(1,1,j)\right\} \left(
1+r_{n}(r,\rho )\text{ }m_{r+\rho }(x_{n},z_{n}\right)  \label{alpha85}
\end{equation}
with $\left| r_{n}(r,\rho )\right| \leq 1-G(x_{n})$ and by convention, $\mu
_{\tau }(1,1,j)=\mu _{\tau }(0,1,j).$ Now by Lemma \ref{lemma42},

\begin{equation}
\underset{n\rightarrow +\infty }{\lim }\mathbb{E}(W_{n}(r)W_{n}(\rho ))\text{
}/\text{ }\tau _{r}(x_{n})\text{ }\tau _{\rho }(x_{n}))=C_{2}(r,\rho )\text{ 
}a(r,\rho ),  \label{alpha90}
\end{equation}
where for $1\leq r<\rho ,$

\begin{equation}
a(r,\rho )=\overset{j=r}{\underset{j=1}{\sum }}a(r-j)\text{ }\mu _{\rho
-r}(1,1,j)  \label{alpha91}
\end{equation}
and $\ for$ $F\in D(\psi _{\gamma }),$

\begin{equation}
C_{2}(\gamma,\rho)=\overset{r}{\underset{j=1}{\Pi}}\left\{ \frac{\gamma +j}{%
\gamma+\rho+j}\right\} \text{ }  \label{alpha92}
\end{equation}
and for $F\in D(\Lambda)$ $\cup D(\varphi),$%
\begin{equation*}
C_{2}(\gamma,\rho)=1\text{.}
\end{equation*}

\subsection{COMPUTATION OF THE COVARIANCE FUNCTION OF $\left\{ w_{n}(r)/%
\protect\tau _{r}(x_{n})\right\} -e(r)B_{n}(\frac{k}{n})$}

\bigskip

It is immediate that 
\begin{equation}
\bigskip \mathbb{E}(W_{n}(r)\text{ }B_{n}(\frac{k}{n}))=(1+r_{n})\frac{n}{k}%
m_{r}(x_{n},z_{n}),where\text{ }\left| r_{n}\right| \leq 1-G(x_{n}).
\label{alpha93}
\end{equation}
Put $\widetilde{W_{n}}(r)=\left\{ W_{n}(r)/\tau _{n}(r)\right\} -e(r)B_{n}(%
\frac{k}{n}).$ Thus,

\begin{equation}
\underset{n\rightarrow +\infty }{\lim }\mathbb{E}(\widetilde{W}_{n}(r)\text{ 
}\widetilde{W}(\rho ))=C_{2}(r,\rho )a(r,\rho )-e(r)-e(\rho )+e(r)e(\rho )
\label{alpha94}
\end{equation}
while

\begin{equation}
\lim_{n\rightarrow \infty }\mathbb{E}\widetilde{W}_{n}(r)^{2}=2C_{1}(r)\text{
}a(r)-2\text{ }e(r)+e(r)^{2}  \label{alpha95}
\end{equation}

We have now finished to compute the covariance functions of W$_{n}(r)/\tau
_{r}(x_{n})$ and $\widetilde{W}_{n}(r)$ via (\ref{alpha82}), (\ref{alpha90}%
), (\ref{alpha94}), and (\ref{alpha95}). These formulas say that $%
w_{n}(r)/\tau _{r}(x_{n})$ is an extremal gaussian process and $\widetilde{W}%
_{n}(r)$ is the reduced process.

These processes deserve a great interest. Particularly, the structure of the
covariance function should be investigated. High order properties, as $%
p\rightarrow +\infty $, must be obtained. All that will be done further.

We just mention now that laws of the iterated logarithm in \cite{gslo2} may
be easily extended to $T_{n}(p)$ as follows :

\begin{equation}
\forall p\geq 1,\underset{n\rightarrow \infty }{\lim }\sup (resp.\inf .)%
\frac{T_{n}(p)-\tau _{p}\left( \overset{\sim }{x}_{n},\overset{\sim }{z_{n}}%
\right) }{\tau _{p}(x_{n})\sqrt{2k\text{ loglogn}}}\overset{a.s}{=}\text{ }%
\rho (p)\text{ \ }(\text{ }resp.-\rho (p)\text{ }),  \label{alpha96}
\end{equation}

\begin{equation}
\forall p\geq 1,\underset{n\rightarrow \infty }{\lim }\sup (resp.\inf .)%
\frac{T_{n}(p)-\tau _{p}(x_{n})}{\tau _{p}(x_{n})\sqrt{2k\text{ loglogn}}}%
\overset{a.s}{=}\text{ }\overset{\sim }{\sigma }(p)\text{ \ }(\text{ }resp.-%
\overset{\sim }{\sigma }(p)),  \label{alpha98}
\end{equation}
where $\overset{\sim }{\sigma }(p)=\sigma ^{2}(p)-e(p)(2-e(p)).$

\end{document}